\documentclass{article}

\usepackage{fullpage, parskip}
\usepackage{amsmath, amsthm, amsfonts}
\usepackage{color}
\usepackage{graphicx}
\usepackage{pdfpages}

\newtheorem{theorem}{Theorem}[section]
\newtheorem{corollary}[theorem]{Corollary}
\newtheorem{lemma}[theorem]{Lemma}
\newtheorem{conjecture}[theorem]{Conjecture}

\newcommand{\mn}{(m,n)}

\title{The Simple Chromatic Number of $(m,n)$-Mixed Graphs}

\date{}
\begin{document}

\maketitle
\begin{center}
	Christopher Duffy\footnote{Corresponding Author: christopher.duffy@usask.ca}\footnote{Research supported by Canada's National Science and Engineering Research Council} \quad \quad Jarrod Pas\\
	\vspace{.15in}
	\begin{small}
		Department of Mathematics and Statistics, University of Saskatchewan, CANADA\\
	\end{small}
\end{center}

\begin{abstract}
	An $(m,n)$-mixed graph generalizes the notions of oriented graphs and edge-coloured graphs to a graph object with $m$ arc types and $n$ edge types. 
	A simple colouring of such a graph is a non-trivial homomorphism to a reflexive target.
	We find that simple chromatic number of complete $(m,n)$-mixed graphs can be found in polynomial time.
	For planar graphs and $k$-trees ($k \geq 3$) we find that allowing the target to be reflexive does not lower the chromatic number of the respective family of $(m,n)$-mixed graphs.
	This implies that the search for universal targets for such families may be restricted to simple cliques.
\end{abstract}

Simple undirected graphs, oriented graphs and edge-coloured graphs may be generalized with a single class of graph objects, $\mn$-mixed graphs.
Such graphs consists of arcs (each of one of $m$ colors) and edges (each of one of $n$ colours).
A homomorphism of an $\mn$-mixed graph $G$ to a $\mn$-mixed graph $H$ is a vertex mapping that preserves arcs, edges and their colours.
This definition of homomorphism, as well as the subsequent definitions of  colouring and chromatic number, generalizes those definitions specific to simple undirected graphs, oriented graphs and edge-colored graphs.

In their work defining $\mn$-mixed graphs,  Ne\v{s}et\v{r}il  and Raspaud showed that a pair of similar results for oriented graphs and edge-coloured graphs were actually special cases of a more general result concerning $\mn$-mixed graphs \cite{NR00}. 
Here we continue this work. 
In particular we find that similar results in the study of reflexive homomorphism for oriented graphs in \cite{S02} and $2$-edge-coloured graphs in \cite{D15} can be generalized to the more general framework of $\mn$-mixed graphs.

A \textit{mixed graph} is a simple undirected graph in which a subset of the edges have been oriented to be arcs. 
An \textit{\mn-mixed graph}, $G = (V, A, E)$, with vertex set $V$, arc set $A$, and edge set $E$,  is a mixed graph 
together with a pair of functions: $c_A: A \to \{1,2,3, \dots, m\}$ and $c_E: E \to \{1,2,3, \dots, n\}$.
When $m = 0$ (resp. $n=0$) it is assumed that the mixed graph used to form the $(0,n)$-mixed graph (resp. $(m,0)$-mixed graph) contains no arcs (resp. edges). 
From this we see that a $(0,1)$-mixed graph is a simple undirected graph, a $(0,k)$-mixed graph is a $k$-edge-colored graph and a $(1,0)$-mixed graph is an oriented graph.

We observe that each $(m,n)$-mixed graph $G$ has an underlying simple undirected graph, which we denote by $U(G)$. 

In discussing arcs and edges of $(m,n)$-mixed graphs we make no distinction in notation between arcs and edges. 
Since each pair of adjacent vertices in $U(G)$ has exactly one adjacency in $G$,  there is no possibility for confusion in the notation $uv$ being used to refer to either an arc from $u$ to $v$ or an edge between $u$ and $v$, as the case may be. 
We say that $uv,xy \in E(U(G))$ have the \emph{same adjacency type} when
\begin{itemize}
	\item $uv,xy \in A$ and $c_A(uv) = c_A(xy) $ or
	\item $uv,xy \in E$ and $c_E(uv) = c_E(xy).$
\end{itemize}

We note that any particular $\mn$-mixed graph need not include all possible adjacency types. 
And so for any fixed $\mn$, we have that every $\mn$-mixed graph $G$ is also an $(m^\prime,n^\prime)$-mixed graph for any $m^\prime \geq m$ and  $n^\prime \geq n$. 

Let $G$ and $H$ be $\mn$-mixed graphs. 
A \emph{homomorphism} of $G$ to $H$ is a function $\phi : V(G) \rightarrow V(H)$ such that for all $uv \in E(U(G))$ we have
\begin{enumerate}
	\item $\phi(u) \neq \phi(v)$;
	\item if $uv \in E(G)$, then $\phi(u)\phi(v) \in E(H)$ and c$_{E(G)}(uv) = c_{E(H)}(\phi(u)\phi(v))$; and
	\item if $uv \in A(G)$, then $\phi(u)\phi(v) \in A(H)$ and c$_{A(G)}(uv) = c_{A(H)}(\phi(u)\phi(v))$.
\end{enumerate}
We call $\phi$ a homomorphism of $G$ to $H$.
Informally, a homomorphism of $G$ to $H$ is a vertex mapping that preserves arcs, edges and their colours.

We write $ \phi:G \rightarrow H$ when there exists a homomorphism $\phi$ of $G$ to $H$, or $G \rightarrow H$ when the name of the function is not important.
An $\mn$-mixed \emph{$k$-colouring} is a homomorphism of $G$ to a target on $k$ vertices.
The \textit{chromatic number} of $G$, denoted $\chi(G)$, is the least integer $t$ such that $G$ admits a homomorphism to a $(m,n)$-mixed graph on $t$ vertices.
For a family of $\mn$-mixed graphs $\mathcal{F}$, we let $\chi(\mathcal{F})$ denote the maximum of $\chi(G)$ taken over all $G \in \mathcal{F}$. If no such maximum exists we say $\chi(\mathcal{F}) = \infty$.

The study of homomorphisms of $(m,n)$-mixed graphs arises from the study of homomorphisms of oriented graphs and homomorphisms of $2$-edge-coloured graphs.
In \cite{NR00}, the authors introduce $(m,n)$-mixed graphs.
Their introduction is motivated by results in \cite{RASO94} for oriented graphs and in \cite{AM98} for $2$-edge-coloured graphs.
In \cite{RASO94}, the authors show that the chromatic number of an oriented graph is bounded above by the function of the acyclic chromatic number of the underlying graph.
In \cite{AM98}, the authors present the same result for $2$-edge-coloured graphs.
In \cite{NR00}, the authors show that these two results are in fact special cases of a more general result for $\mn$-mixed graphs.

More recently in \cite{BDS17}, we have seen the unification of results for oriented graphs and $2$-edge-coloured graphs in the classification of cliques and some related computational complexity problems for $(m,n)$-mixed graphs.
We continue this trend here, and find results in the area of simple homomorphisms for $(m,n)$-mixed graphs that generalize known results for oriented graphs and $2$-edge-coloured graphs.

A \emph{simple homomorphism} of $G$ to $H$ is a function $\phi : V(G) \rightarrow V(H)$ such that either $|V(G)| = 1$, or
\begin{enumerate}
	\item there exists $x,y \in V(G)$ such that $\phi(x) \neq \phi(y)$;
	\item if $uv \in E(G)$, then $\phi(u)\phi(v) \in E(H)$ and c$_{E(G)}(uv) = c_{E(H)}(\phi(u)\phi(v))$; and
	\item if $uv \in A(G)$, then $\phi(u)\phi(v) \in A(H)$ and c$_{A(G)}(uv) = c_{A(H)}(\phi(u)\phi(v))$.
\end{enumerate}

We write $ \phi:G \rightarrow_s H$ when there exists a simple homomorphism, $\phi$, of $G$ to $H$, or $G \rightarrow_s H$ when the name of the function is not important. 
One may consider a simple homomorphism to be a non-trivial homomorphism to a target with a loop of each adjacency type at each vertex.
We define the \emph{simple chromatic number}, denoted $\chi_s$, analogously to the chromatic number.
We use the term \emph{simple $k$-colouring} to refer to a simple homomorphism of $G$ to an $\mn$-mixed graph with $k$ vertices.
We call an subgraph  $S$ of $G$ \emph{monochromatic} with respect to $\phi$  when each vertex of $S$ has the same image with under $\phi$.

When restricted to the class of oriented graphs (i.e., $(1,0)$-mixed graphs), simple homomorphism is akin to homomorphism to reflexive targets in which at least two vertices in the target have a pre-image.
In her PhD thesis, Smol\'inkov\'a introduced and studied various aspects of simple colourings of oriented graphs \cite{S02}. 
Notably, she showed $\chi_s(\mathcal{P}) = \chi(\mathcal{P})$ for $\mathcal{P}$, the family of orientations of planar graphs.
Various aspects of the study of homomorphism to reflexive digraphs  have been examined \cite{G08b, G08, M12}.
Simple colourings of edge-coloured graphs were considered for the first time in \cite{D15}.

Our work proceeds as follows. 
In Section \ref{sec:Complete} we fully characterize  minimum simple colourings of complete $\mn$-mixed graphs. 
Our results imply the existence of a polynomial-time algorithm to compute the simple chromatic number of  such $\mn$-mixed graphs.
We find constructions for $\mn$-mixed graphs on $k \geq 5$ vertices so that the resulting $\mn$-mixed graph has simple chromatic number exactly $k$.
In Section \ref{sec:minimum} we study  \emph{optimally simply colourable} families. 
That is, those families for which $\chi_s(\mathcal{F}) = \chi(\mathcal{F}$).
We generalize results given in \cite{D15} and \cite{S02}.
These results suggest that for such families, which include planar graphs, that the search for universal targets can be restricted to simple cliques, those $\mn$-mixed graph that have chromatic number equal to their order.

Herein we assume that all $\mn$-mixed graphs are simple.
That is, between every pair of vertices there is at most one adjacency, and that no vertex is adjacent to itself.
For all other graph theoretic definitions and notation we refer the reader to \cite{bondy}.

\section{The Simple Chromatic Number of $(m,n)$-mixed Graphs}\label{sec:Complete}

We begin by classifying those $\mn$-mixed graphs with chromatic number  $2$.\\

\begin{theorem} \label{kedge:2col}
	An $\mn$-mixed graph $G$ with $U(G) = \Gamma$ has $\chi_s(G) = 2$ if and only if there exists a partition $V(G) = X \cup Y$ so that for all $xy,x^\prime y^\prime \in E(\Gamma)$ with $x,x^\prime \in X$ and $y, y^\prime \in Y$ we have
	\begin{itemize}
		\item $xy, x^\prime y^\prime \in E(G)$ and $c_E(xy)= c_E(x^\prime,y^\prime)$;
		\item $xy, x^\prime y^\prime \in A(G)$ and $c_A(xy)= c_A(x^\prime,y^\prime)$; or
		\item $yx, y^\prime x^\prime \in A(G)$ and $c_E(yx)= c_E(y^\prime,x^\prime)$.
	\end{itemize}
\end{theorem}

\begin{proof} 	
By definition  we have $\chi_s(G) = 2$ if and only if there exists $\phi:G \to_s H$ so that $U(H)$ is $2K_1$ or $K_2$.
Such a partition $V(G) = X \cup Y$ is equivalent to a homomorphism of $G$ to $H$.
\end{proof}

For an $\mn$-mixed graph $G$, let $E_i$ (respectively $A_j$) be the set of edges $e$ (respectively arcs) so that $c_E(e) = i$ (respectively, $c_A(e) = j$).
Let $G[E - E_i]$ (respectively $G[A - A_j]$) be the subgraph of $G$ formed by removing all edges of $E_i$ (respectively, all arcs of $A_j$) from $G$.\\

\begin{corollary}\label{cor:edgeCut}
	An $\mn$-mixed graph $G$  has $\chi_s(G) = 2$ if and only if there exists a minimal edge cut $E^\prime$ of $U(G)$ so that 
	\begin{itemize}
		\item there exists $1 \leq i \leq n$ so that $e \in E_i$ for all $e \in E^\prime$; or
		\item there exists $1 \leq j \leq m$ so that $uv \in A_j$ for all $uv \in E^\prime$ and there is a partition of the components of $U[G - E^\prime]$ in to two sets, $X$ and $Y$, so that every arc in $G$ between $x \in X$ and $y \in Y$ has its head at $y$.
	\end{itemize} 
\end{corollary}

Corollary \ref{cor:edgeCut} trivially implies  a graph $G$ has $\chi_s(G) = 2$ if and only if  $G$ has an edge cut or $G$ is disconnected.
Theorem \ref{kedge:2col} generalizes the classification of oriented graphs with simple chromatic number $2$ found in \cite{G08b} and \cite{S02}, and the classification of $2$-edge-coloured graphs with simple chromatic number $2$ found in \cite{D15}.\\

\begin{corollary}\label{cor:poly2}
	For any $(m,n)$-mixed graph $G$, it can be decided in polynomial time if $\chi_s(G) = 2$.
\end{corollary}

\begin{proof}
If $U(G)$ is not connected or $|V(G)| \leq 2$, then $\chi_s(G) \leq 2$.
And so we may assume $U(G)$ is connected and $G$ has at least $3$ vertices.
By Corollary \ref{cor:edgeCut} it suffices to decide if there exists a minimal edge cut $E^\prime$ of $U(G)$ whose removal leaves a pair of subgraphs $X$ and $Y$ so that for all $xy,x^\prime y^\prime \in E(\Gamma)$ with $x,x^\prime \in X$ and $y, y^\prime \in Y$ we have
\begin{itemize}
	\item $xy, x^\prime y^\prime \in E(G)$ and $c_E(xy)= c_E(x^\prime,y^\prime)$;
	\item $xy, x^\prime y^\prime \in A(G)$ and $c_A(xy)= c_A(x^\prime,y^\prime)$; or
	\item $yx, y^\prime x^\prime \in A(G)$ and $c_E(yx)= c_E(y^\prime,x^\prime)$.
\end{itemize}

For each of the $n$ edge types, this amounts to deciding whether the subgraph of $U(G)$ induced by removing all the edges of a particular colour is connected, which can be decided in polynomial time.

Assume  $G[E - E_i]$ is connected for all $1 \leq i \leq n$.
If $U(G[A - A_j])$ is connected ($1 \leq j \leq m$), then no subset of the arcs of colour $j$ can form $E^\prime$ in the statement of Corollary \ref{cor:edgeCut}.
And so $G$ does not admit a simple $2$-colouring.
Otherwise, consider the case where $U(G[A - A_j])$ is not connected, with components $\mathcal{C}_j = \{C_1, C_2, \dots, C_t\}$ for some $1 \leq j \leq m$.
Let $H_j$ be the digraph with vertex set $\mathcal{C}_j$, where there is the arc from $C_k$ to $C_\ell$ ($1 \leq k,\ell \leq t$) when there is an arc $c_kc_\ell \in A(G)$ for any $c_k \in C_k$ and $c_\ell \in C_\ell$.
By the construction of $H_j$, such an arc in $G$ must have $c_A(c_kc_\ell) = j$.

\emph{Claim: $\chi_s(G) > 2$ if and only if  $H_j$ is strongly connected for all $1 \leq j \leq m$}.\\
If $H_j$ is not strongly connected, then there exists a partition $X_j \cup Y_j$ of $\mathcal{C}_j$ so that all of the arcs,  between $X_j$ and $Y_j$ have their tail in $X_j$.
By construction, this set of arcs is a minimal edge cut in $U(G)$.
By Corollary \ref{cor:edgeCut}, we have $\chi_s(G) = 2$.\\

Assume $H_j$ is strongly connected for all $1 \leq j \leq m$.
Therefore for all partitions $X_j\cup Y_j$ of $V(H_j)$, there is an arc from $X_j$ to $Y_j$ and an arc from $Y_j$ to $X_j$.
Therefore there is no set of arcs with colour $j$ that can be used to satisfy Corollary \ref{cor:edgeCut}

The result follows from the claim by observing the that strong connectedness can be decided in polynomial time for digraphs.
\end{proof}

Let $G$ be a $\mn$-mixed graph. If $uv,wv \in E(U(G))$, then we say $u$ and $w$ \emph{agree} on $v$ when
\begin{itemize}
	\item $uv,vw \in E(G)$ and $c_E(uv) = c_E(vw)$; 
	\item $uv, wv \in A(G)$ and $c_A(uv) = c_A(wv)$; or
	\item $vu, vw \in A(G)$ and $c_A(vu) = c_A(vw)$.
\end{itemize}

Otherwise we say $v$ is \emph{between} $u$ and $w$.
Note that the definitions of \emph{agree} and \emph{between} follow the usual definitions for oriented graphs.

If $v$ is between $u$ and $w$,  we observe that if $\phi$ is a simple colouring of $G$ such that $\phi(u) = \phi(w)$, then $\phi(v) = \phi(u)$. 
Following \cite{S02}, we say $C \subset V(G)$ is \emph{convex} if for any pair $u,w \in C$ there is no $v \in V(G) - C$ such that $v$ is between $u$ and $w$.  
For $N \subseteq V(G)$, the \emph{convex hull} of $N$ is the smallest convex set of vertices of $G$ that has $N$ as a subset. 
It is easily verified that this set is well-defined.
We denote this set $conv(N)$.\\

\begin{lemma}\label{lem:ConvSame}
	Let $c$ be a simple  colouring of a $\mn$-mixed graph $G$ and consider $N \subseteq V(G)$ such that for all $u \in N$, $c(u) = i$. For all $x \in conv(N)$ we have $c(x) = i$.
\end{lemma}

\begin{proof}
	Consider a vertex $x \in conv(N)$ and let $N^\prime$ be the largest subset of $conv(N)$ such that $x \notin N^\prime$ and for all $y,z \in N^\prime$ such that if there is a vertex $w \neq x$ between $y$ and $z$, then $w \in N^\prime$. We proceed by induction on the cardinality of $N^\prime$. If $|N^\prime| = 2$, then, since $N^\prime$ is largest, $x$ is between the two vertices in $N^\prime$ and so $c(x) = i$. 
	
	Assume now that $|N^\prime| = k > 2$. Since $N^\prime$ is largest, there exists a pair of vertices $y,z \in N^\prime$ such that $x$ is between $y$ and $z$. If $c$ is a simple colouring of $G$, then by induction $c(y) = c(z) = c(v)$ for all $v \in N^\prime$. Since $x$ is between $y$ and $z$, it must also be $c(x) = c(u)$.
\end{proof}

Though, in general, simple homomorphisms do not compose, surjective simple homomorphisms do compose.
We formalize this observation in the following lemma.\\

\begin{lemma} \label{lem:ComposeInject}
	If $\phi: G \to_s H$ and $\beta: H \to_s J$ are surjective, then $G \to_s J$.
\end{lemma}

\begin{proof}
	Let $\phi: G \to_s H$ and $\beta: H \to_s J$.
	Since both $\phi$ and $\beta$ preserve arcs, edges, and their colours, we have that $\beta \circ \phi: G \to J$ preserves arcs, edges, and their colours.
	To show $\beta \circ \phi$ is a simple homomorphism, it suffices to show that there exists $u,v \in V(G)$ such that $\beta \circ \phi(u) \neq \beta \circ \phi(v)$.
	Since $\beta$ is surjective, for every $y \neq z \in V(J)$ there exists $y_H, z_H \in V(H)$ such that $y_H = \beta^{-1}(y)$  $z_H = \beta^{-1}(z)$ and $y_H \neq z_H$.
	Similarly, there exists $y_G, z_G \in V(G)$ such that $y_G = \phi^{-1}(y_H)$, $z_G = \phi^{-1}(z_H)$ and $y_G \neq z_G$. 
	Therefore $\beta \circ \phi(y_G) \neq \beta \circ \phi(z_G)$, as required.
\end{proof}

\begin{theorem} \label{thm:ChromNumberInject}
	If $\phi: G \to_s H$ is surjective, then $\chi_s(G) \leq \chi_s(H)$.
\end{theorem}

\begin{proof}
	This follows directly from Lemma \ref{lem:ComposeInject} and the definition of simple chromatic number.
\end{proof}

We call $G$ an $\mn$-mixed graph \emph{complete} when $U(G) = K_{|V(G)|}$.\\

\begin{corollary}\label{cor:chiSurj}
	If $G$, an $\mn$-mixed graph, has $\chi_s(G) = k$, then there exists a complete $\mn$ mixed graph $T$ so that $\phi :G \to_s T$ is surjective and $|V(T)| = k$.
\end{corollary}

Let $T_3$ be the transitive tournament on three vertices. 
Despite being an orientation of a complete graph, we have $\chi_s(T_3) = 2$.
Recall that an $(m,n)$-mixed clique is an $(m,n)$-mixed graph $G$ so that $\chi(G) =|V(G)|$.
The family of $(m,n)$-mixed cliques is classified as follows.\\

\begin{theorem}\cite{BDS17} \label{thm:BDS}
	An $\mn$-mixed graph $G$ is an  $(m,n)$-mixed clique if and only if for every $u,v \in V(G)$ either $uv \in E(U(G))$ or there is a vertex $z$ such that $z$ is between $u$ and $v$.
\end{theorem} 

We define an \emph{$(m,n)$-mixed simple clique} analogously, and find the following classification.\\

\begin{theorem}\label{thm:isaClique}
	An $\mn$-mixed graph $G$  is a  $(m,n)$-mixed simple clique if and only if $conv(\{u,v\}) = |V(G)|$ for all $u,v\in V(G)$.
\end{theorem}

\begin{proof} 
	Let $G$ be an $(m,n)$-mixed simple clique.
	The claim holds when $|V(G)| \leq 2$.
	Assume 	$|V(G)| \geq 3$ and there exists $u,v \in V(G)$ so that $conv(\{u,v\}) \subset V(G)$.
	Let $H$ be the $(m,n)$-mixed graph formed by identifying all vertices of $conv(\{u,v\})$ into a single vertex.
	There exists a surjective simple homomorphism $\phi: G \to_s H$.
	By Theorem \ref{thm:ChromNumberInject}, we have $\chi_s(G) \leq \chi_s(H)$.
	However, since $H$ has fewer vertices than $G$, we have $\chi_s(H) < |V(G)| = \chi_s(G)$, a contradiction.

	Assume $conv(\{u,v\}) = V(G)$ for all $u,v\in V(G)$.
	In particular we have $conv(uv) = V(G)$ for all $u,v\in E(U(G))$.
	Therefore every simple colouring of $G$ is also an $(m,n)$-mixed colouring of $G$.
	Thus it suffices to show that $G$ is an $(m,n)$-mixed clique.
	Since $conv(\{u,v\}) = V(G)$ for all $u,v\in V(G)$, for all non-adjacent pairs $x,y \in V(G)$, there is a vertex between $x$ and $y$.
	The result follows by Theorem \ref{thm:BDS}.
\end{proof}

\begin{corollary}\label{cor:CompleteClique}
	A complete $(m,n)$-mixed graph $G$ is an $(m,n)$-mixed clique if and only if  $conv(\{uv\}) = |V(G)|$ for all $uv\in E(U(G))$.
\end{corollary}

We show that the convex hull of a pair of vertices of an arbitrary $\mn$-mixed graph can be efficiently computed, and so the problem of deciding if a $\mn$-mixed graph is an $(m,n)$-mixed simple clique is Polynomial.

Let $G$ be a $\mn$-mixed graph.
Let $X \subseteq V(G)$.
Let $N^{X}_0 = X$.
For $i > 0$, if $N^{X}_i$ exists, let $B^{X}_i \subseteq V(G) - N^{X}_i$ be the set of vertices $x$ such that $x$ is between a pair of vertices in $N^{X}_i$.
If $B^{X}_i \neq \emptyset$, then let $N^{X}_{i+1} = N^{X}_i \cup B^{X}_i$.
Otherwise let $N^{X} = N^{X}_i$.\\

\begin{theorem}\label{thm:ConvexHullIterative}
Let $G$ be an $\mn$-mixed graph. For every $X \subseteq V(G)$ we have	$N^{X} = conv(X)$
\end{theorem}

\begin{proof}
	Observe the claim holds when $|X| = 1$.
	Assume $|X| \geq 2$.
	By definition, $N^{X}$ is convex and $X \subseteq N^{X}$.
	If  $N^{X} \neq conv(X)$, then by definition of $conv(X)$, we have $|N^{X}| > |conv(X)|$.
	Therefore there is some least index $i \geq 0$ so that $B^{X}_i$ contains a vertex $x$ where $x \in N^{X}$, but $x \notin conv(X)$.
	By definition, $x$ is between a pair of vertices in $N^{X}_{i}$.
	We reach a contradiction by observing that $N^{X}_{i} \subseteq N^{X}$ and $N^{X}$ is convex, but $x \notin N^{X}$.
\end{proof}

\begin{corollary} \label{cor:ConvexSubsets}
	For $G$,  an $\mn$-mixed graph, and  $X \subseteq X^\prime \subseteq V(G)$, we have $Conv(X) \subseteq Conv(X^\prime)$.\\
\end{corollary}

\begin{corollary}
	Deciding if a $\mn$-mixed graph is an $(m,n)$-mixed simple clique is Polynomial.
\end{corollary}

\begin{proof}
	The result follows from Theorems \ref{thm:isaClique} and \ref{thm:ConvexHullIterative} and by observing that $B_i$ in the construction of $conv(X)$ can be constructed in polynomial time.
\end{proof}

Though computing the simple chromatic number of an oriented graph is NP-hard \cite{S02}, we find that the simple chromatic number of any complete $(m,n)$-mixed graph can be computed in polynomial time.
To show this we require the following result.\\

\begin{theorem} \label{thm:notComplete}
	Let $G$ be a complete $\mn$-mixed graph with $\chi_s(G) > 2$.
	We have  $conv(\{u,v\}) \neq V(G)$ if and only if $c(u) = c(v)$ in every minimum simple colouring of $G$.
\end{theorem}

\begin{proof}
	Let  $G$ be a complete $\mn$-mixed graph with $\chi_s(G) > 2$.
	If $c(u) = c(v)$ in any simple colouring of $G$, then $conv(\{u,v\}) \neq V(G)$.
	
	Assume now that $conv(\{u,v\}) \neq V(G)$, but $c(u) \neq c(v)$.
	By Lemma \ref{cor:chiSurj}, there exists an $\mn$-mixed clique $T$ so that $\phi :G \to_s T$ is surjective and $V(T) = \chi_s(G)$.
	By Theorem \ref{thm:isaClique}, $conv(\{c(u), c(v)\}) = V(T)$.
	Let $Z = \{z \in V(G) | c(u) \neq c(z) \neq c(v)\}$.
	As $c(z) \in conv(\{c(u),c(v)\})$ and $c$ is surjective, we have  $Z \subset conv(\{u,v\})$.
	Since $\chi_s(G) > 2$, $Z$  is non-empty.
	Consider $z^\prime \in Z$.
	By Theorem \ref{thm:isaClique}, $conv(\{c(z^\prime), c(v)\}) = V(T)$.
	Let $V^\prime = \{v^\prime \in V(G) | c(v^\prime)= c(v)\}$.
	As $c(v) \in conv(\{c(u),c(z^\prime)\})$ and $c$ is surjective we have $V^\prime  \subseteq conv(\{z^\prime, u\})$.
	Thus $V^\prime \subseteq conv(\{u,v\})$.
	Similarly, $\{u^\prime \in V(G) | c(u^\prime) = c(u)\} \subseteq conv(\{u,v\})$.
	Therefore $conv(\{u,v\})  = V(G)$, a contradiction.
	Therefore $c(u) = c(v)$.

\end{proof}

\begin{corollary} \label{cor:ConvEdges}
	Let $G$ be a complete $\mn$-mixed graph so that $\chi_s(G) > 2$ and let $c$ be a minimum colouring of $G$.
	If $c(u) \neq c(v)$, then $Conv\{u,v\} = V(G)$.\\
\end{corollary}

Theorem \ref{thm:notComplete} implies  minimum simple colourings of complete $\mn$-mixed graphs are unique up to permutation of the labels of the colours.\\

\begin{corollary}
	Let $G$ be a complete $\mn$-mixed graph with $\chi_s(G) > 2$. Up to the labelling of the colour classes, there is a unique minimum colouring of $G$.
\end{corollary}

\begin{proof}
	Let $c$ be a minimum colouring of $G$.
	Assume $c$ is not unique.
	Therefore there exists $u,v \in V(G)$ and a minimum colouring $c^\prime$ of $G$  such that $c(u) = c(v)$, but $c^\prime(u) \neq c^\prime(v)$.
	Since  $c(u) = c(v)$, we have $conv(\{u,v\}) \neq V(G)$.
	However, since $c^\prime(u) \neq c^\prime(v)$, we have $conv(\{u,v\}) = V(G)$, a contradiction.\\
\end{proof}

\begin{theorem} \label{thm:CompletePoly}
	Let $G$ a complete $(m,n)$-mixed graph. For any fixed $k \geq 1$ it can be decided in polynomial time if $\chi_s(G) \leq k$
\end{theorem}

\begin{proof}
	The only complete  $(m,n)$-mixed graph with simple chromatic number $k=1$ has a single vertex.
	For $k=2$, the result follows from Corollary \ref{cor:poly2}.
	Otherwise for fixed $k \geq 3$, we proceed by induction on $\nu = V(G)$, noting the claim is true when $\nu = 3$.
	Consider $G$ with $\nu = \ell$ and  $uv \in U(G)$.
	By Theorem \ref{thm:isaClique}, it can be decided in polynomial time if $G$ is a simple clique.
	If $G$ is not a simple clique, then by Corollary \ref{cor:ConvEdges} there exists $u$ and $v$ so that $conv(\{u,v\}) \neq V(G)$.
	As noted in the proof of Corollary \ref{cor:poly2}, $conv(\{u,v\})$ can be constructed in polynomial time.
	Let $G^\prime$ be the $\mn$-mixed graph formed from $G$ by identifying the vertices of $conv(\{u,v\})$ in to a single vertex.
	By Lemma \ref{lem:ConvSame} and Theorem \ref{thm:notComplete}, we have $\chi_s(G) = \chi_s(G^\prime)$.
	The result follows by induction.
\end{proof}

The complexity of deciding if an oriented graph has simple chromatic number at most $k$ (for fixed $k$) is NP-complete for all $k \geq 5$, and polynomial otherwise \cite{S02}. 
As such, the complexity of deciding if an $\mn$-mixed graph with $m \neq 0$ has simple chromatic number at most $k$ (for fixed $k$) is NP-complete when $k \geq 5$.
However, by Theorem \ref{thm:CompletePoly}, this problem is Polynomial when restricted to the class of complete $(m,n)$-mixed graphs.
Given the close relationship between oriented colourings and colourings of $2$-edge-coloured graphs, it is reasonable to expect this decision problem to be NP-complete when $m=0$ and $n \geq 2$. 
This remains to be verified.

%
%
%

We close our discussion on the simple chromatic number of $(m,n)$-mixed graphs by providing constructions for $(m,n)$-mixed simple cliques.
We use these constructions to show that for every $\mn \neq (0,1)$ and every $k \geq 5$ that there is a $\mn$-mixed simple clique with $k$ vertices.
Such a theorem, of course, would be trivial in the study of graph colourings -- the complete graph on $k$ vertices always has chromatic number $k$.
However, the observation that there exist complete $(m,n)$-mixed graphs whose simple chromatic number is strictly less than their order makes the above theorem of interest for simple colourings.
Consider, for example, the family of tournaments on four vertices.
It is easily checked that every tournament on four vertices has simple chromatic number at most $3$.
And so this implies  there are no oriented graphs with simple chromatic number $4$.
We find a similar behaviour amongst the family of $2$-edge-coloured complete graphs on $3$ vertices.
Our result implies  for all $k\neq 3,4$, and all $\mn \neq (0,1)$  there is an $\mn$-mixed graph with chromatic number $k$.

Let $\Gamma$ be an additive group and let $S \subseteq \Gamma$.
The \emph{Cayley digraph} $D(\Gamma,S)$ has $V(D) = \Gamma$ and $uv \in A(D)$ when $u-v \in S$.
If $S$ contains no pair $\{x,-x\}$, then $G(\Gamma,S)$ is necessarily an oriented graph.
The \emph{Cayley graph}, $G(\Gamma,S)$ has $V(G) = \Gamma$ and $uv \in E(G)$ when $u-v \in S$.
For a Cayley graph $G(\Gamma,S)$ we may assume $S$ is closed under the additive inverse.\\

For $x,n \in \mathbb{Z}_{+}$, let $[x]_n \in \{0,1,2,3\dots n-1\}$ be the result of reducing $x$ modulo $n$.\\

\begin{theorem}\label{lem:CayleyDigraphClique}
Let $n \equiv 1 \pmod 2$, $\Gamma = \mathbb{Z}_n$ and  $S = \{2, n-1\} \cup \{x \in V(D) | x \equiv 0 \pmod 4 \}$. The Cayley digraph $D(\Gamma,S)$ is a simple oriented clique for all $n \geq 5$.
\end{theorem}
 
\begin{proof}
	By Theorem \ref{thm:isaClique} is suffices to show that $Conv\{u,v\} = V(G)$ for all $u,v \in V(G)$.
	Observe that that vertices $i, [i+1]_n,[i+2]_n$ induce a copy of the oriented $3$-cycle for all $0 \leq i \leq n$.
	Therefore $\{i, [i+1]_n,[i+2]_n\} \subseteq conv(\{i,[i+2]_n\})$ and $\{i, [i+1]_n,[i+2]_n\} \subseteq conv(\{i,[i+1]_n\})$.
	Thus  by Corollary \ref{cor:ConvexSubsets}, if $\{i,[i+1]_n\} \subseteq N$, or $\{i,[i+2]_n\} \subseteq N$, then $conv(N) = V(G)$ for any $1 \leq i \leq n$.
	
	Since $G$ is vertex transitive, it suffices to assume $u = 0$. 
	By the remarks above it suffices to show either $1 \in conv(\{0,v\})$ or $2 \in conv(\{0,v\})$. 
	Further we may assume $v \neq -2,-1,1,2$
	
	\emph{Case I $v \equiv 0 \pmod 4$}:
	If $v\neq 4$, then observe that the vertex $4$ is between $0$ and $v$.
	Therefore $4 \in conv(\{0,v\})$.
	Observe now that the vertex $2$ is between $0$ and $4$.
	Therefore $2 \in  conv(\{0,v\})$.
	By the remarks above and  Corollary \ref{cor:ConvexSubsets}  we have $conv(\{0,v\}) = V(G)$.
	
	\emph{Case II $v \equiv 2 \pmod 4$:}
	Observe that the vertex $1$ is between $0$ and $v$.
	Therefore $1 \in conv(\{0,v\})$.
	By the remarks above and  Corollary \ref{cor:ConvexSubsets}  we have $conv(\{0,v\}) = V(G)$.
	
	\emph{Case III $v \equiv 1 \pmod 4$:}
	If $n \equiv 1 \pmod 4$, and $v \neq n-4$, then
	observe that the vertex $n-4$ is between $0$ and $v$.
	Therefore $n-4 \in conv(\{0,v\})$.
	Observe that the vertex $n-2$ is between $0$ and $n-4$.
	Therefore $n-2 \in conv(\{0,v\})$.
	By the remarks above and  Corollary \ref{cor:ConvexSubsets}  we have $conv(\{0,v\}) = V(G)$.
	
	Otherwise, $n \equiv 3 \pmod 4$.
	Observe that the vertex $n-2$ is between $0$ and $v$.
	Therefore $n-2 \in conv(\{0,v\})$.
	By the remarks above and  Corollary \ref{cor:ConvexSubsets}  we have $conv(\{0,v\}) = V(G)$.

	\emph{Case IV $v \equiv 3 \pmod 4$}:
	If $n \equiv 1 \pmod 4$, then observe that the vertex $n-2$ is between $0$ and $v$
	Therefore $n-2 \in conv(\{0,v\})$.
	By the remarks above and  Corollary \ref{cor:ConvexSubsets}  we have $conv(\{0,v\}) = V(G)$.

	Otherwise, $n \equiv 3 \pmod 4$.
	Observe that the vertex $n-1$ is between $0$ and $v$
	Therefore $n-1 \in conv(\{0,v\})$.
	By the remarks above and  Corollary \ref{cor:ConvexSubsets}  we have $conv(\{0,v\}) = V(G)$.
	
	Therefore $conv(\{0,v\}) = V(G)$ for all $v \in V(G)$.
	
\end{proof}

\begin{lemma}\label{lem:CayleyGraphClique}
	Let  $\Gamma = \mathbb{Z}_n$ for $n \geq 5$ and  $S = \{1\} \cup \{x \in V(G) \;|\; x \equiv 1,2\pmod 4 \}$. The Cayley $2-$edge-coloured graph $G(\Gamma,S\cup -S)$ with $c_E(uv) = 1$ if and only if $|v-u| = 1$  is a simple clique.
\end{lemma}

\begin{proof}
		By Theorem \ref{thm:isaClique} is suffices to show that $Conv\{u,v\} = V(G)$ for all $u,v \in V(G)$.
		Observe the vertices $i, [i+1]_n,[i+2]_n,[i+3]_n$ induce a  $2$-edge-coloured simple clique for all $0 \leq i \leq n$.
		Let $C_i = \{i, [i+1]_n,[i+2]_n,[i+3]_n\}$ for $1 \leq i \leq n$.
		Thus if for any $u,v \in V(G)$, if any $2-$element subset of $C_i$ (for a fixed $1 \leq i \leq n$) is a subset of $Conv\{u,v\}$, we have $Conv\{u,v\} = V(G)$.
		
		Observe the $2$-edge-coloured Cayley graph $G(\Gamma,S)$  is vertex transitive.
		Thus it suffices to assume $u = 0$.
		By the previous remarks we may assume $v \neq -3,-2,-1,1,2,3$.
		Observe that either $v+1$ or $v-1$ is between  $0$ and $v$ for all $5 \leq v \leq n-5$.
		Therefore  $\{v,v+1\} \subset conv(\{0,v\})$ or $\{v,v-1\} \subset conv(\{0,v\})$.
		The result follows from the previous remarks and Corollary \ref{cor:ConvexSubsets}.
\end{proof}

For $n \geq 3$, let $H_n$ be the $2$-edge-coloured graph with vertex set
$V(H_n) = \{x_0,x_1,\dots x_{n-1}\} \cup \{y_0,y_1,\dots y_{n-1}\}$ constructed from $K_{n,n}$ by adding edges so that each of  $x_0,x_1,\dots x_{n-1}$ and $y_0,y_1,\dots y_{n-1}$ are cycles.
We complete the construction of $H_n$ by letting $c_E(e) = 1$ when  $e = x_ix_{[i+1]_n}$, $e = y_iy_{[i+1]_n}$ ($0 \leq i \leq n-1$)
or $e = x_jy_j$ ($0 \leq j \leq n-1$).
Otherwise we let $c_E(e) = 2$.\\

\begin{lemma}\label{lem:HnClique}
	$H_n$ is a simple $2-$edge-coloured simple clique for all $n \geq 3$.
\end{lemma}

\begin{proof}
By Theorem \ref{thm:isaClique}, it suffices to check that $conv(\{u,v\}) = V(H_n)$.
Let $X = \{x_0, x_1, \dots, x_{n-1}\}$ and $Y = \{y_0, y_1, \dots, y_{n-1}\}$.
We proceed by induction on $n$, noting the claim is true by inspection when $n = 3$.
By symmetry we may assume $u = x_{n-1}$.
We proceed in cases.

\emph{Case I: $u = x_{n-1}, v = y_{n-1}$.} 
By construction, $x_{n-1}$ and $y_{n-1}$ disagree on $x_{n-2}$ and $y_{n-2}$. 
Therefore $\{x_{n-2},y_{n-2}\} \subset conv(\{u,v\})$.
By induction, $conv(\{x_{n-2},y_{n-2}\}) = X \cup Y - \{x_{n-1}, y_{n-1}\}$.
Therefore $conv(\{u,v\}) = V(G_n)$.

\emph{Case II: $u = x_{n-1}, v = x_{n-2}$}. 
The vertices $x_{n-1}$ and $x_{n-2}$ disagree on $y_{n-1}$.
The result follows by Case I.

\emph{Case III: $u = x_{n-1}, v = x_j$, $(1 \leq j < n-1$ and $j \neq n-1)$}. 
The vertices $x_{n-1}$ and $x_{j}$ disagree on $y_{n-1}$.
The result follows by Case I.

\emph{Case IV: $u = x_{n-1}, v = y_j$, $(1 \leq i,j \leq n-1$ and $j \neq n-1)$}. 
The vertices $x_{n-1}$ and $x_{j}$ disagree on $x_{n-2}$.
The result follows by Case III.
\end{proof}

For $n\geq 3$, let $G_n$ be the oriented graph with vertex set
$V(G_n) = \{x_0,x_1,\dots x_{n-1}\} \cup \{y_0,y_1,\dots y_{n-1}\}$ constructed from $K_{n,n}$ by adding arcs so that each of  $x_0,x_1,\dots x_{n-1}$ and $y_0,y_1,\dots y_{n-1}$ are directed cycles and orienting the edges so that $x_iy_j \in A(G_n)$ for $1 \leq i \leq j \leq n$ and $y_j x_i \in A(G_n)$ for $1 \leq j < i \leq n$.\\

\begin{lemma} \label{lem:GnClique}
	$G_n$ is a simple oriented clique for all $n \geq 3$.
\end{lemma}

\begin{proof}
This proof follows similarly to that of Lemma \ref{lem:HnClique}.
\end{proof}

\begin{theorem}
	For every $\mn \neq (0,1)$ and every $k \geq 5$ that there is an $\mn$-mixed clique with $k$ vertices.
\end{theorem}

\begin{proof}
	Observe that an oriented graph is an $\mn$-mixed graph for all $m \geq  1$  and that a $2$-edge-coloured graph is an $\mn$-mixed graph for all $n \geq 2$.
	The result now follows from  Lemmas \ref{lem:CayleyDigraphClique}, \ref{lem:CayleyGraphClique} and \ref{lem:GnClique}.
\end{proof}

In \cite{MNP75} the authors show that almost every tournament is a simple oriented clique.
In many cases, the constructions above give  oriented or $2$-edge-coloured simple cliques whose underlying graphs are not complete.
In \cite{BDS17} the authors show for every $\mn \neq (0,1)$ that almost every $\mn$-mixed graph is an $\mn$-mixed clique.
And so we conjecture the following. \\

\begin{conjecture}
	For every $\mn \neq (0,1)$, almost every $\mn$-mixed graph is a simple $\mn$-mixed clique.
\end{conjecture}

\section{Optimally Simply Colourable Families}\label{sec:minimum}
Recall that a family of $(m,n)$-mixed graphs $\mathcal{F}$ is \emph{optimally simply colourable} when $\chi_s(\mathcal{F}) = \chi(\mathcal{F})$.
One can see that any family $\mathcal{F}$ with $\chi_s(\mathcal{F}) = \infty$ is optimally simply colourable.
This is the case for orientations of bipartite graphs \cite{S02}.
However, there are also examples of optimally simply colourable families with finite chromatic number.
We show this is the case for $\mathcal{T}_k$, the family of $\mn$-mixed graphs whose underlying graphs are partial $k$-trees, for fixed $k \geq 3$, and for $\mathcal{P}$, the family of $\mn$-mixed graphs whose underlying graphs are planar.
In each case we show that there exists a $\mn$-mixed graph whose oriented chromatic number is maximum across all $\mn$-mixed graphs in the family, but that has the same simple and oriented chromatic number. 

Let $\mathcal{H}$ be the family of $(1,0), (2,0)$ and $(0,2)$ mixed graphs given in Figure \ref{fig:planarFig} together with those formed from reversing the orientation of every arc or swapping edge/arc colour $1$ with edge/arc colour $2$.
Ornamentation on the arc/edge indicates arc/edge colour differences.
\\

\begin{figure}
	\begin{center}
		\includegraphics[width = .75\linewidth]{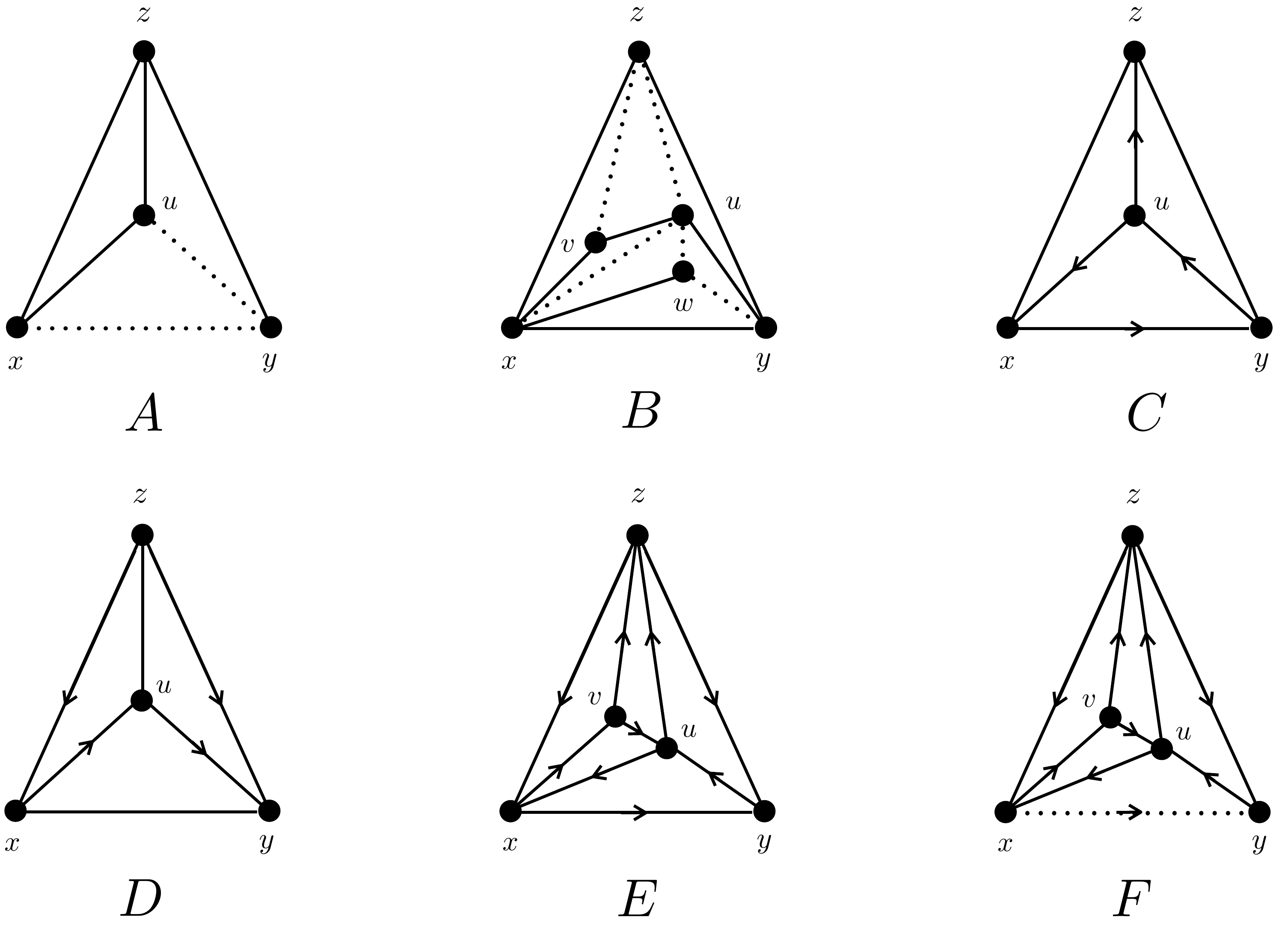}
	\end{center}
	\caption{The family $\mathcal{H}$ in the proof of Lemma \ref{lem:1or3}.}
	\label{fig:planarFig}
\end{figure}

\begin{lemma}\label{lem:1or3}
	Let $G$ be an $\mn$-mixed graph that contains $H$, a member of $\mathcal{H}$. For every simple colouring $\phi$ of $G$, the set $\{\phi(x), \phi(y), \phi(z)\}$ has cardinality either $1$ or $3$.
\end{lemma}

\begin{proof}
	Let $G$ be an $\mn$-mixed graph that contains $H\in \mathcal{H}$. 
	Let $\phi$ be a simple colouring of $G$. 
	Assume $|\{\phi(x), \phi(y), \phi(z)\}| \neq 3$. 
	There are six possibilities for $H$.
	
	\emph{Case I: $H=A$, or $H=C$ or $H=D$.}
	
	\emph{Case I.i: $\phi(x) = \phi(z)$.}\\
	If $\phi(x) = \phi(z)$, then $\phi(x) = \phi(y)$ since $y$ is between $x$ and $z$.
	Thus $\phi(x) = \phi(y) = \phi(z)$.
	
	\emph{Case I.ii: $\phi(y) = \phi(z)$.}\\
	The proceeds similarly to the previous case.
	
	\emph{Case I.iii: $\phi(x) = \phi(y)$.}\\
	If $\phi(x) = \phi(y)$, then $\phi(u) = \phi(y)$ as $u$ is between $x$ and $y$. 
	Therefore $\phi(z) = \phi(y)$, as $z$ is between $u$ and $y$. 
	Thus $\phi(x) = \phi(y) = \phi(z)$.
	
	Thus if $H=A$, or $H=C$ or $H = D$ and $|\{\phi(x), \phi(y), \phi(z)\}| \neq 3$, then $|\{\phi(x), \phi(y), \phi(z)\}| = 1$.
	
	\emph{Case II: $H = B$}. \\
	\emph{Case II.i: $\phi(x) = \phi(y)$.}\\
	If $\phi(x) = \phi(y)$, then $\phi(x) = \phi(w)$, as $w$ is between $x$ and $y$.
	Therefore $\phi(u) = \phi(y)$, as $u$ is between $w$ and $y$.
	Therefore  $\phi(z) = \phi(y)$, as $z$ is between $u$ and $y$.
	Thus $\phi(x) = \phi(y) = \phi(z)$.
	
	\emph{Case II.ii: $\phi(x) = \phi(z)$.}\\
	If $\phi(x) = \phi(z)$, then $\phi(x) = \phi(v)$, as $v$ is between $x$ and $z$.
	Therefore $\phi(u) = \phi(v)$, as $u$ is between $v$ and $x$.
	Therefore $\phi(w) = \phi(u)$, as $w$ is between $u$ and $x$.
	Therefore $\phi(y) = \phi(w)$, as $y$ is between $w$ and $x$.
	Thus $\phi(x) = \phi(y) = \phi(z)$.
	
	\emph{Case II.iii: $\phi(y) = \phi(z)$.}\\
	The proceeds similarly to the previous case.
	
	Thus if $H=B$ and $|\{\phi(x), \phi(y), \phi(z)\}| \neq 3$, then $|\{\phi(x), \phi(y), \phi(z)\}| = 1$.
	
	\emph{Case III: $H = E$ or $H=F$}. \\
	\emph{Case III.i: $\phi(x) = \phi(y)$.}\\
	This case follows similarly to \emph{Case I.iii}
	
	\emph{Case III.ii: $\phi(x) = \phi(z)$.}\\
	If $\phi(x) = \phi(z)$, then $\phi(x) = \phi(v)$, as $v$ is between $x$ and $z$.
	Therefore $\phi(u) = \phi(v)$, as $u$ is between $v$ and $x$.
	Therefore $\phi(y) = \phi(u)$, as $y$ is between $u$ and $x$.
	Thus $\phi(x) = \phi(y) = \phi(z)$.
	
	\emph{Case III.iii: $\phi(y) = \phi(z)$.}\\
	This follows similarly to the previous case by exchanging the role of $u$ and $v$.
	
	Thus if $H = E$ or $H=F$ and $|\{\phi(x), \phi(y), \phi(z)\}| \neq 3$, then $|\{\phi(x), \phi(y), \phi(z)\}| = 1$.	
	This completes the proof.
\end{proof}

We note that the case $H=E$ appears in \cite{S02}. We include it here for completeness.

Let $\mathcal{H}_{xyz}$ be the family of graphs formed from those in $\mathcal{H}$ by taking the subgraphs induced by $x,y$, and $z$ for each $H \in \mathcal{H}$. \\

\begin{theorem}\label{thm:planarTheorem}
	For all $\mn \neq (0,1)$, the family of $\mn$-mixed planar graphs is optimally simply colourable.
\end{theorem}

\begin{proof}
	For fixed $\mn \neq (0,1)$,  let $\mathcal{P}$ be the family of $\mn$-mixed planar graphs.
	Since it is clear  $\chi_s(\mathcal{P}) \leq \chi(\mathcal{P})$, it suffices to show  $\chi_s(\mathcal{P}) \geq \chi(\mathcal{P})$.
	Let $G$ be a $\mn$-mixed graph such that $U(G)$ is maximally planar and $\chi(G) = \chi(\mathcal{P})$. 
	We note that such a $G$ must exist as adding arcs/edges to a $\mn$-mixed graph cannot decrease its chromatic number. 
	Let $\mathcal{C}$ be the set of simple colourings $c$ of $G$ using no more than $\chi_s(\mathcal{P})$ colours that have a monochromatic arc/edge.
	If $\chi_s(\mathcal{P}) < \chi(\mathcal{P})$, then for each $c\in C$, there at least one triangular face $F_c$ so that exactly two of the vertices of this face are assigned the same colour by $c$.
	Such faces cannot be monochromatic directed $3$-cycles, nor can such a face contain arcs/edges of three distinct colours.
	Therefore each $F_c$  is isomorphic to some $H_c \in  \mathcal{H}_{xyz}$, up to the labels of the colours of the edges/arcs.
	
	We construct an $\mn$-mixed planar graph $G^\prime$ that contains $G$ as a proper subgraph.
	For each $c \in \mathcal{C}$ and each $F_c$ add the necessary vertices, arcs and edges to form $H^\prime_c$, the graph from $\mathcal{H}$ used to form $H_c$.
	Observe $\chi(G^\prime) \geq \chi(G) = \chi(\mathcal{P})$. 
	Since each $H \in \mathcal{H}$ is planar, it follows  $G^\prime \in \mathcal{P}$.
	Therefore $\chi(G^\prime) = \chi(\mathcal{P})$.
	
	Let $c^\prime$ be a simple colouring of $G^\prime$ using  $\chi_s(G^\prime)$ colours.
	Consider $c^\prime|_G$, the simple colouring produced by restricting $c^\prime$ to vertices of $G$.
	If $c^\prime|_G$ is not a  proper colouring, then there must be some triangular face $F$ of $G$ so that exactly two of the vertices of this face are assigned the same colour by $c^\prime|_G$.
	Since $c^\prime|_G$ uses no more than $\chi_s(\mathcal{P})$ colours, the existence of $F$ implies  $c^\prime|_G \in \mathcal{C}$.
	
	However, by the construction of $G^\prime$ and Lemma \ref{lem:1or3} the vertices on $F$ either all are assigned the same colour or all assigned distinct colours.
	This contradicts the existence of $F$.
	Since no such $F$ exists, it must be  $c^\prime|_G$ is in fact a  colouring of $G$.
	Since $\chi(G) = \chi(\mathcal{P})$, it follows that $c^\prime|_G$ uses $\chi(\mathcal{P})$ colours.
	This implies  $c^\prime$ uses at least $\chi(\mathcal{P})$ colours. 
	Thus $\chi_s(G^\prime) \geq \chi(\mathcal{P})$.
	Therefore $\chi_s(\mathcal{P}) \geq \chi(\mathcal{P})$, as required.
\end{proof}

We proceed similarly to show the  the family of $\mn$-mixed partial $k$-trees ($k \geq 3$) is optimally simply colourable.

\begin{theorem}\label{thm:TreeTheorem}
	For all $\mn \neq (1,0)$ and all $k \geq 3$,  the family of \mn-mixed partial $k$-trees is optimally simply colourable.
\end{theorem}

\begin{proof}
	For fixed $\mn \neq (0,1)$ and $k\geq 3$,  let $\mathcal{T}_k$ be the family of  partial $\mn$-mixed $k$-trees.
	Since it is clear  $\chi_s(\mathcal{T}_k) \leq \chi(\mathcal{T}_k)$, it suffices to show  $\chi_s(\mathcal{T}_k) \geq \chi(\mathcal{T}_k)$.
	Let $G$ be a $\mn$-mixed graph such that $U(G)$ is a $k$-tree and $\chi(G) = \chi(\mathcal{T}_k)$. 
	We note that such a $G$ must exist, as adding arcs/edges to a $\mn$-mixed graph cannot decrease its chromatic number. 
	Let $\mathcal{C}$ be the set of simple colourings $c$ of $G$ using no more than $\chi_s(\mathcal{T}_k)$ colours that have a monochromatic arc/edge.
	Recall that as $G$ is $k$-tree it constructed via sequence of cliques of order $k+1$: $V_1, V_2, \dots, V_\ell$.
	If $\chi_s(\mathcal{T}_k) < \chi(\mathcal{T}_k)$, then for each $c\in C$, there is at least one clique $V_i = \{v_1,v_2,\dots v_{k+1}\}$ so that $2 \leq |\{ \phi(u_1), \phi(u_2), \dots, \phi(u_{k+1})\}| \leq k$.

	Therefore each such clique $V_i$ contains a triple of vertices $u,v,w$ so that $|\{c(u),c(v),c(w)\}|=2$.
	Such triples cannot induce monochromatic directed $3$-cycles, nor can such a triple of vertices induce a subgraph that contains contain arcs/edges of three distinct colours.
	For each $c \in \mathcal{C}$ let $B_c$ the set of such triples.
We see then  that for each $c \in \mathcal{C}$ each element $\{u,v,w\}$ of $B_c$ induces a subgraph $H^c_{u,v,w}$ that is isomorphic to an element of $\mathcal{H}_{xyz}$.

	We construct an $\mn$-mixed partial $k$-tree $G^\prime$ that contains $G$ as a proper subgraph.
	For each $c \in \mathcal{C}$ and each element $\{u,v,w\}$ of $B_c$ we add the necessary vertices, arcs and edges to form the graph from $\mathcal{H}$ used to form $H^c_{u,v,w}$.
	As $G$ is a proper subgraph of $G^\prime$ we have $\chi(G^\prime) \geq \chi(G) = \chi(\mathcal{T}_k)$. 
	Since each $H \in \mathcal{H}$ is a partial $k$-tree, it follows  $G^\prime \in \mathcal{T}_k$.
	Therefore $\chi(G^\prime) = \chi(\mathcal{T}_k)$.
	
	Let $c^\prime$ be a simple colouring of $G^\prime$ using  $\chi_s(G^\prime)$ colours.
	Consider $c^\prime|_G$, the simple colouring produced by restricting $c^\prime$ to vertices of $G$.
	If $c^\prime|_G$ is not a proper colouring, then there must be some clique $V = \{v_1,v_2,\dots v_{k+1}\}$ so that $2 \leq |\{ \phi(u_1), \phi(u_2), \dots, \phi(u_k)\}| \leq k$.
	Therefore such a clique contains a triple of vertices $u,v,w$ so that $|\{c(u),c(v),c(w)\}|=2$.
	Since $c^\prime|_G$ uses no more than $\chi_s(\mathcal{P})$ colours, the existence of $\{u,v,w\}$ implies  $c^\prime|_G \in \mathcal{C}$.
	
	However, by the construction of $G^\prime$ and Lemma \ref{lem:1or3} the vertices $\{u,v,w\}$ either all are assigned the same colour or all assigned distinct colours.
	This contradicts the existence of $\{u,v,w\}$ and thus the existence of $V$.
	Since no such $V$ exists, it must be  $c^\prime|_G$ is a  proper colouring of $G$.
	Since $\chi(G) = \chi(\mathcal{T}_k)$, it follows $c^\prime|_G$ uses $\chi(\mathcal{T}_k)$ colours.
	This implies  $c^\prime$ uses at least $\chi(\mathcal{T}_k)$ colours. 
	Thus $\chi_s(G^\prime) \geq \chi(\mathcal{T}_k)$.
	Therefore $\chi_s(\mathcal{T}_k) \geq \chi(\mathcal{T}_k)$, as required.
\end{proof}

The method in the proofs of Theorems \ref{thm:planarTheorem} and \ref{thm:TreeTheorem} also yields the following result:

\begin{theorem}\label{thm:planartree}
	For all $\mn \neq (1,0)$   the family of  \mn-mixed planar graphs with tree width $3$ is simply optimally colourable.
\end{theorem}

One can easily observe that family of $(m,n)$-mixed trees is not simple optimally colourable -- by Theorem \ref{cor:edgeCut} each such tree has simple chromatic number $2$.
This leaves open the question of simple optimal colourability for $2$-trees. Here we show that this is not the case for orientations of $2$-trees and $2$-edge-coloured $2$-trees.

\begin{theorem}
	If $G$ is an oriented $2$-tree, then $\chi_s(G) \leq 3$.
\end{theorem}

\begin{proof}
	We show by induction  every orientation of a $2$-tree admits a simple homomorphism to the directed cycle on three vertices, noting such a homomorphism exists of both the transitive triple and the  directed cycle on three vertices.
	Let $H$ be the directed cycle on three vertices with arc set $A(H) = \{x_1x_2, x_2x_3, x_3x_1\}$.
	Consider $T$, an orientation of a $2$-tree on $n> 3$ vertices.
	As $T$ is an orientation of $2$-tree, $U(T)$ necessarily has a vertex of degree $2$, say $z$.
	Let $z_1$ and $z_2$ be the neighbours of $z$ in $T$ so that $z_1z_2 \in A(T)$.
	By induction there exists $\phi: T-\{z\} \to_s H$.
	As $H$ is vertex transitive  and arc transitive, we may assume $\phi(z_1) = x_1$ and $\phi(z_2) \in \{x_1,x_2\}$.
	The table below gives an extension of $\phi$ to include $z$ based on  $\phi(z_2)$ and the direction of the arcs between $z$ and $z_1$ and between $z$ and $z_2$. 
	In the table below, $+$ denotes that relevant vertex is an out-neighbour of $z$ and $-$ denotes that the relevant vertex is an in-neighbour of $z$.
	\begin{center}
$	\begin{array}{c|c|c|c}
		z_1 & z_2 & \phi(z_2) & \phi(z)\\
		\hline
		+ & + & x_1 &  x_2\\
		+ & - & x_1 &  x_1\\
		- & + & x_1 &  x_1\\
		- & - & x_1 &  x_3\\
		+ & + & x_2 & x_2 \\
		+ & - & x_2 & x_1\\
		- & + & x_2 & x_3\\
		- & - & x_2 &  x_1
	\end{array}$
 	\end{center}
\end{proof}

In \cite{OP08}, Ochem and Pinlou show the chromatic number of the family of oriented $2$-trees is $7$.

\begin{corollary}
	The family of orientations of partial $2$-trees is not optimally simply colourable
\end{corollary}

\begin{theorem}
	If $G$ is a $2$-edge-coloured $2$-tree, then $\chi_s(G) \leq 5$.
\end{theorem}

\begin{proof}
		We show by induction that every $2$-edge-coloured $2$-tree admits a simple homomorphism to the complete $2$-edge-coloured graph whose red edges induce a copy of $C_5$, noting that such a homomorphism exists for each of the four $2$-edge-coloured $2$-trees on three vertices.
		Let $H$ be the complete $2$-edge-coloured graph on five vertices whose red edges form the cycle $C = x_1,x_2,x_3,x_4,x_5$.
		Consider $T$, a $2$-edge-coloured $2$-tree on $n > 3$ vertices.
		As $U(T)$ is a $2$-tree, $U(T)$ necessarily has a vertex of degree $2$, say $z$.
		Let $z_1$ and $z_2$ be the neighbours of $z$ in $T$.
		By induction there exists $\phi: T-\{z\} \to_s H$.
		Notice $H$ is vertex transitive.
		Also observe $\overline{C_5} \cong C_5$ and $C_5$ is edge transitive.
		Thus we may assume, without loss of generality that $\phi(z_1) = x_1$ and $\phi(z_2) \in \{x_1,x_2\}$.
		The table below gives an extension of $\phi$ to include $z$ based on  $\phi(z_2)$ and the colour the edges $zz_1$ and $zz_2$. 
		In the table below, $r$ denotes the existence of a red edge between $z$ and the relevant vertex and $b$ denotes the existence of a blue edge between $z$ and the relevant vertex.
				
		\begin{center}
			$	\begin{array}{c|c|c|c}
			z_1 & z_2 & \phi(z_2) & \phi(z)\\
			\hline
			r & r & x_1 & x_2  \\
			r & b & x_1 &  x_1 \\
			b & r & x_1 &  x_1\\
			b & b & x_1 & x_3  \\
			r & r & x_2 & x_1 \\
			r & b & x_2 & x_4 \\
			b & r & x_2 & x_2\\
			b & b & x_2 & x_3 
			\end{array}$
		\end{center}
\end{proof}

\begin{corollary}
	The family of $2$-edge-coloured partial $2$-trees is not optimally simply colourable.
\end{corollary}

\begin{proof}
	The $2$-edge-coloured $2$-tree given in Figure \ref{fig:2tree} is a clique and thus has chromatic number $6$. 

\end{proof}

\begin{figure}	
\begin{center}
		\includegraphics[width= 0.25\linewidth]{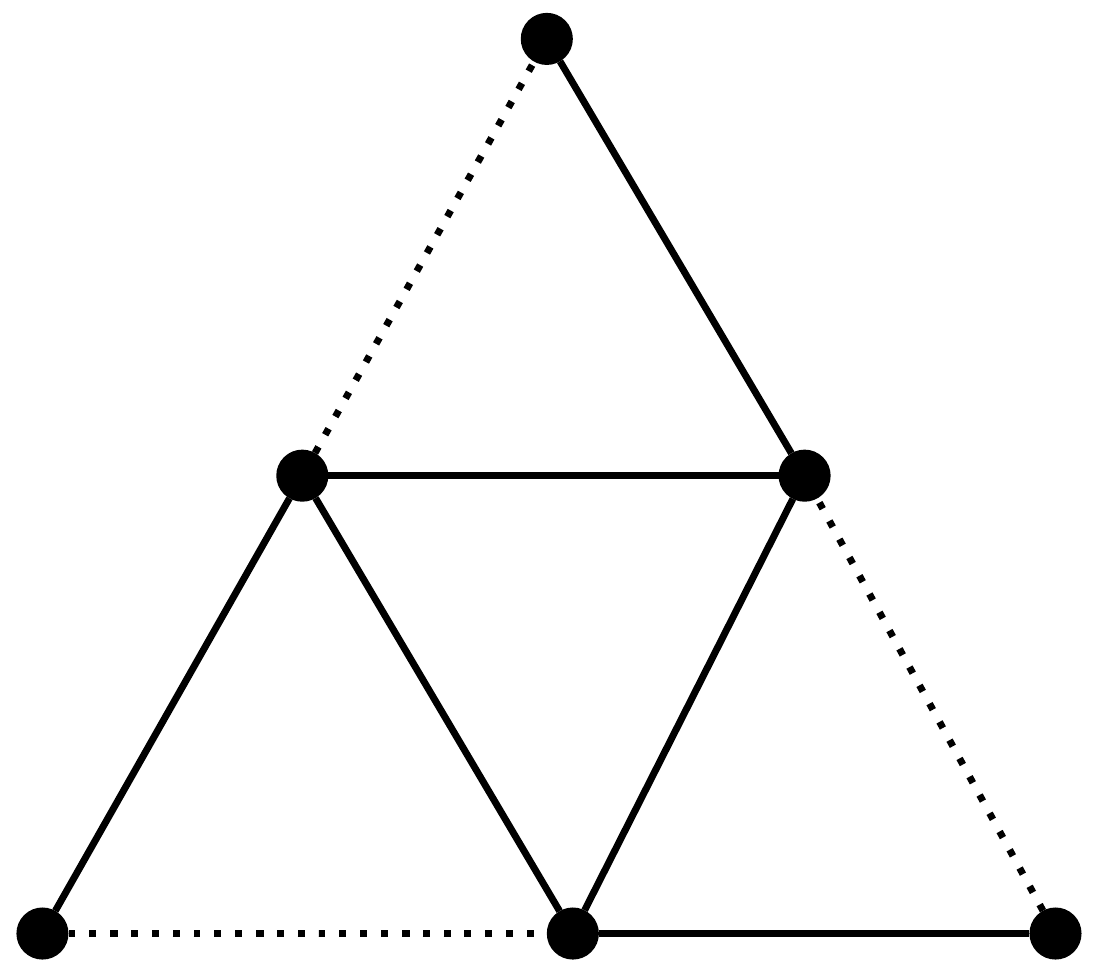}
\end{center}
	\caption{A $2$-edge-coloured $2$-tree with chromatic number $6$.}
	\label{fig:2tree}
\end{figure}

\section{Conclusion}
A common technique to find an upper bound for the chromatic number of a family of $(m,n)$-mixed graphs $\mathcal{F}$ is to a find an  $(m,n)$-mixed graph $H$ such that $F \to H$ for all $F \in \mathcal{F}$.
Such a target $(m,n)$-mixed graph is called \emph{universal} for the class $\mathcal{F}$.
Theorems \ref{thm:planarTheorem} and \ref{thm:TreeTheorem} suggest that one may restrict the search for universal targets to the class of simple cliques for $\mn$-mixed planar graphs and $\mn$-mixed $k$-trees.
In \cite{NR00} the authors construct $H^{\mn}_k$, a universal target for the class of $(m,n)$-mixed graphs  whose underlying graphs have acyclic chromatic number at most $k$, for fixed $(m,n)$ and $k$.
Since planar graphs have acyclic chromatic number at most $5$ \cite{Bo79}, and so $H^{\mn}_5$ is a universal target for the class of $(m,n)$-mixed planar graphs.
If $H^{\mn}_k$ were not a simple clique, then Theorem \ref{thm:planarTheorem} together with a simple colouring of $H^{\mn}_k$ could be used to improve the upper bound for the chromatic number of $(m,n)$-mixed planar graphs using Lemma \ref{lem:ComposeInject}.
In particular, one could improve the long-standing upper bound of $80$ for the oriented chromatic number of planar graphs \cite{RASO94}.
However, one can verify by computer that $H^{(1,0)}_5$ is indeed a simple oriented clique.

\section*{Acknowledgements}
The authors thank an anonymous reviewer for their observations leading to Theorem \ref{thm:planartree} and the structure of Section \ref{sec:minimum}.
\bibliographystyle{abbrv}
\bibliography{references.bib}

\begin{thebibliography}{10}

\bibitem{AM98}
N.~Alon and T.~Marshall.
\newblock Homomorphisms of {Edge}-{Colored} {Graphs} and {Coxeter} {Groups}.
\newblock {\em Journal of Algebraic Combinatorics}, 8(1):5--13, 1998.

\bibitem{BDS17}
J.~Bensmail, C.~Duffy, and S.~Sen.
\newblock Analogues of cliques for {$(m, n)$}-colored mixed graphs.
\newblock {\em Graphs and Combinatorics}, 33(4):735--750, 2017.

\bibitem{bondy}
J.~Bondy and U.~Murty.
\newblock {\em Graph Theory}.
\newblock Number 244 in Graduate {Texts} in {Mathematics}. Springer, 2008.

\bibitem{Bo79}
O.~V. Borodin.
\newblock On acyclic colorings of planar graphs.
\newblock {\em Discrete Mathematics}, 25(3):211--236, 1979.

\bibitem{D15}
C.~Duffy.
\newblock {\em Homomorphisms of {$(j,k)$}-mixed graphs}.
\newblock PhD thesis, University of Bordeaux/University of Victoria, 2015.

\bibitem{G08b}
A.~Gupta, P.~Hell, M.~Karimi, and A.~Rafiey.
\newblock Minimum cost homomorphisms to reflexive digraphs.
\newblock {\em LATIN 2008: Theoretical Informatics}, pages 182--193, 2008.

\bibitem{G08}
G.~Gutin, A.~Rafiey, and A.~Yeo.
\newblock Minimum cost homomorphisms to semicomplete bipartite digraphs.
\newblock {\em SIAM Journal on Discrete Mathematics}, 22(4):1624--1639, 2008.

\bibitem{M12}
M.~Mar{\'o}ti and L.~Z{\'a}dori.
\newblock Reflexive digraphs with near unanimity polymorphisms.
\newblock {\em Discrete Mathematics}, 312(15):2316--2328, 2012.

\bibitem{MNP75}
V.~Muller, J.~Ne\v{s}et\v{r}il, and J.~Pelant.
\newblock Either {Tournaments} or {Algebras?}
\newblock {\em Discrete Mathematics}, 11(1):37 -- 66, 1975.

\bibitem{NR00}
J.~Ne\v{s}et\v{r}il and A.~Raspaud.
\newblock {Colored} {Homomorphisms} of {Colored} {Mixed} {Graphs}.
\newblock {\em Journal of Combinatorial Theory, Series B}, 80(1):147 -- 155,
  2000.

\bibitem{OP08}
P.~Ochem and A.~Pinlou.
\newblock Oriented colorings of partial 2-trees.
\newblock {\em Information Processing Letters}, 108(2):82--86, 2008.

\bibitem{RASO94}
A.~Raspaud and E.~Sopena.
\newblock {Good} and {Semi}-{Strong} {Colorings} of {Oriented} {Graphs}.
\newblock {\em Information Processing Letters}, 51:171--174, 1994.

\bibitem{S02}
P.~Smol\'{i}kov\'{a}.
\newblock {\em Simple {Colorings} and {Simple} {Homomorphisms}}.
\newblock PhD thesis, Charles University, 2002.

\end{thebibliography}

\end{document}